\documentclass{article}

\usepackage{amsmath,amssymb,amsthm}
\usepackage{epsfig}
\usepackage{subfigure}

\allowdisplaybreaks

\newcommand{\BigO}{\mathrm{O}}

\newcommand{\BigOmega}{\mathrm{\Omega}}
\newcommand{\littleo}{\mathrm{o}}

\newcommand{\heading}[1]{\vspace{1ex}\par\noindent{\bf #1}}

\newlength{\fparwidth}
\newcommand\framedpar[1]{\begin{center}\framebox{~\begin{minipage}
{\fparwidth}\vspace{1mm}#1\vspace{1mm}
\end{minipage}~}\end{center}}

\def\immediateFigure#1{%
\smallskip\begin{center}#1\end{center}\smallskip }

\newcommand{\iepsfig}[1]
{\immediateFigure{\mbox{\psfig{file=#1.eps}}}}

\newtheorem{theorem}{Theorem}
\newtheorem{lemma}{Lemma}
\newtheorem{proposition}{Proposition}

\begin{document}

\title{Minimum and maximum against $k$ lies}%

\date{\today}

\author{
Michael Hoffmann\\
Institute of Theoretical Computer Science,\\ 
ETH Zurich, Switzerland \\
E-mail: \texttt{hoffmann@inf.ethz.ch}
\and
Ji\v{r}\'{\i} Matou\v{s}ek\\
Department of Applied Mathematics and\\
Institute for Theoretical Computer Science (ITI),\\
Charles University, Prague, Czech Republic \\
E-mail: \texttt{matousek@kam.mff.cuni.cz}
\and
Yoshio Okamoto
\thanks{%
Supported by Global COE Program
  ``Computationism as a Foundation for the Sciences''
  and Grant-in-Aid for Scientific Research from Ministry
  of Education, Science and Culture, Japan, and
  Japan Society for the Promotion of Science.
}\\
Graduate School of Information Science and Engineering, \\
Tokyo Institute of Technology, Japan \\
E-mail: \texttt{okamoto@is.titech.ac.jp}
\and
Philipp Zumstein\\
Institute of Theoretical Computer Science, \\
ETH Zurich, Switzerland \\
E-mail: \texttt{zuphilip@inf.ethz.ch}
}

\maketitle

\begin{abstract}
  A neat 1972 result of Pohl asserts that
$\lceil 3n/2 \rceil -2$ comparisons are sufficient, and also necessary
  in the worst case, for finding both the minimum and the maximum
  of an $n$-element totally ordered set. The set
   is accessed via an oracle for pairwise comparisons.
  More recently, the problem has been studied
  in the context of the R\'enyi--Ulam liar games,
  where the oracle may give up to $k$ false answers.
For large $k$, an upper bound due to
 Aigner shows that $(k+\BigO(\sqrt{k}))n$ 
 comparisons suffice.
 We improve on this by providing an algorithm with
at most $(k+1+C)n+\BigO(k^3)$ comparisons for some constant $C$.
The known lower bounds
are of the form $(k+1+c_k)n-D$, for some constant $D$, where $c_0=0.5$,
$c_1=\frac{23}{32}= 0.71875$, 
and $c_k=\BigOmega(2^{-5k/4})$ as $k\to\infty$.
\end{abstract}

\section{Introduction}

We consider an $n$-element set $X$ with an unknown total ordering $\leq$.
The ordering can be accessed via an oracle that,
given two elements $x,y\in X$, tells us whether $x<y$ or $x>y$.
It is easily seen that  the minimum element of $X$ can be found
using  $n-1$ comparisons.  This is optimal in the sense that $n-2$ 
comparisons are not enough to find the minimum element in the worst case.

One of the nice little surprises in computer science is that
if we want to find \emph{both} the minimum and the maximum,
we can do significantly \emph{better} than finding the minimum
and the maximum separately. Pohl~\cite{Pohl72} proved
that  $\lceil 3n/2 \rceil -2$ is the optimal number of comparisons 
for this problem ($n\geq 2$). The algorithm first partitions the
elements of $X$ into pairs and makes a comparison in each pair.
The minimum can then be found among the ``losers'' of these comparisons,
while the maximum is found among the ``winners.''

Here we consider the problem of determining both the minimum and the maximum
in the case where the oracle is not completely
reliable: it may sometimes give a false answer, but
only at most $k$ times during the whole computation, 
where $k$ is a given parameter.

We refer to this model as \emph{computation against $k$ lies}.
Let us stress that we admit repeating the same query
to the oracle several times, and each false answer
counts as a lie. This seems to be the most sensible 
definition---if repeated queries were not allowed,
or if the oracle could always give the wrong answer to
a particular query, then the minimum cannot be determined.

So, for example, if we repeat a given query $2k+1$ times,
we always get the correct answer by majority vote.
Thus, we can simulate any algorithm with a reliable oracle,
asking every question $2k+1$ times, but for the problems
considered here, this is not a very efficient way, as we will see.

The problem of finding both the minimum and the maximum
against $k$ lies was investigated by Aigner~\cite{Aigner97},
who proved that $(k+\BigO(\sqrt{k}))n$ comparisons always 
suffice.\footnote{Here
and in the sequel, $\BigO(.)$ and $\BigOmega(.)$ hide only
absolute constants, independent of both $n$ and~$k$.}
We improve on this  as follows.

\begin{theorem}\label{thm:1}
  There is an algorithm that finds both the minimum
and the maximum among $n$ elements against $k$ lies using
at most $(k+1+C)n+\BigO(k^3)$ comparisons,
 where $C$ is a constant.
\end{theorem}

Our proof yields the constant $C$ reasonably small
(below $10$, say,
at least if $k$ is assumed to be sufficiently large), 
but we do not try to optimize it.

\heading{Lower bounds. } 
The best known lower bounds for the number of comparisons
necessary to determine both the minimum and the maximum against $k$ lies
have the form $(k+1+c_k)n-D$, where $D$ is a small constant and
the $c_k$ are as follows:
\begin{itemize}
\item $c_0=0.5$, and this is the best possible.
This is the result of Pohl~\cite{Pohl72} for a truthful oracle mentioned above.
\item $c_1=\frac{23}{32}=0.71875$, and this is again tight.
This follows from a recent work by 
Gerbner, P\'alv\"olgyi, Patk\'os, and Wiener~\cite{GPPW}
who determined
the optimum number of comparisons for $k=1$ up to
a small additive constant: it lies
between $\lceil \frac{87}{32}n \rceil -3$
and $\lceil \frac{87}{32}n \rceil +12$. This proves a conjecture
of Aigner~\cite{Aigner97}.
\item $c_k=\BigOmega(2^{-5k/4})$ for all $k$, as was shown by
 Aigner~\cite{Aigner97}.
\end{itemize}

The optimal
constant $c_1=\frac{23}{32}$ indicates that obtaining precise answers
for $k>1$ may be difficult.

\heading{Related work. } The problem of determining the minimum
alone against $k$ lies was resolved
 by Ravikumar, Ganesan, and Lakshmanan~\cite{RGL87}, who proved
that finding the minimum against $k$ lies can be performed
by using at most $(k+1)n-1$ comparisons, and this is optimal in
the worst case.

The problem considered in this paper belongs to the area
of \emph{searching problems against lies} and, in a wider
context, it is an example of ``computation in the presence of errors.''
This field has a rich history and beautiful results.
A prototype problem, still far from completely solved,
is the \emph{R\'enyi--Ulam liar game} from the 1960s, where one wants to 
determine an unknown integer $x$ between
$1$ and $n$, an oracle provides comparisons of $x$ with
specified numbers, and it may give at most $k$ false answers.
We refer to the surveys by Pelc~\cite{Pelc} and
by Deppe~\cite{Deppe} for more information.

\section{A simple algorithm}

Before proving Theorem~\ref{thm:1}, we explain a simpler
algorithm, which illustrates the main ideas but yields
a weaker bound. We begin with formulating a generic algorithm,
with some steps left unspecified. Both the simple algorithm
in this section and an improved algorithm in the next
sections are instances of the generic algorithm.

\framedpar{
\begin{center}\bf The generic algorithm \end{center}

\begin{enumerate}
\item \label{alg1:partition}
  For a suitable integer parameter $s=s(k)$, 
  we arbitrarily partition the considered $n$-element set
  $X$ into $n/s$ groups $X_1,\ldots,X_{n/s}$ of size $s$ 
each.\footnote{If $n$ is not divisible by $s$, we can form an extra
group smaller than $s$
and treat it separately, say---we will not bore the reader with
the details.}
\item \label{alg1:sort}
  In each group $X_i$, we find the minimum $m_i$ and the maximum $M_i$.
The method for doing this is left unspecified in the generic
algorithm.
\item \label{alg1:find}
   We find the minimum of $\{m_1,\ldots,m_{n/s}\}$
against $k$ lies, and independently, we find the
maximum of $\{M_1,M_2,\ldots,M_{n/s}\}$
against $k$ lies.
\end{enumerate}
}

The correctness of the generic algorithm is clear, provided that
Step~\ref{alg1:sort} is implemented correctly. 
Eventually, we set $s:=k$ in the simple and in the improved algorithm. 
However, we keep $s$ as a separate parameter, because the choice $s:=k$ is in a sense accidental.

In the simple algorithm we implement Step~\ref{alg1:sort} as follows.

\medskip

\framedpar{
\begin{center}\bf Step \ref{alg1:sort} in the simple algorithm
\end{center}
\begin{enumerate}
  \renewcommand{\theenumi}{\ref{alg1:sort}.\arabic{enumi}}
\item \label{again} (Sorting.) We sort the elements of $X_i$ 
by an asymptotically optimal sorting algorithm,
say mergesort, using $\BigO(s\log s)$ comparisons,
and ignoring the possibility of lies. Thus,
we obtain an ordering $x_1,x_2,\ldots,x_s$
of the elements of $X_i$ such that \emph{if} all queries during the sorting
have been answered correctly, \emph{then} 
$x_1< x_2< \cdots < x_s$.
If there was at least one false answer, we make no assumptions,
except that the sorting algorithm does not crash and outputs
some ordering.
\item \label{verify} (Verifying the minimum and maximum.)
For each $j=2,3,\ldots,s$, we query the oracle
$k+1$ times with the pair $x_{j-1},x_{j}$. If any of these
queries returns the answer $x_{j-1}>x_{j}$, we \emph{restart}:
We go back
to Step~\ref{again} and repeat the computation for the group $X_i$
from scratch. Otherwise, if all the answers are $x_{j-1}<x_{j}$,
we proceed with the next step. 
\item\label{allright}
We set $m_i:=x_1$ and $M_i:=x_s$.
\end{enumerate}
}
\medskip

\begin{lemma}[Correctness]
The simple algorithm always correctly computes
the minimum and the maximum against $k$ lies.
\end{lemma}

\begin{proof}
We note that once the processing of the group $X_i$
in the above algorithm reaches Step~\ref{allright}, then
$m_i=x_1$ has to be the minimum. Indeed, for every other element
$x_j$, $j\ge 2$, the oracle has answered $k+1$ times that
$x_j>x_{j-1}$, and hence $x_j$ cannot be the minimum.
Similarly, $M_i$ has to be the maximum,
and thus the algorithm is always correct.
\end{proof}

Actually, at Step~\ref{allright} we can be sure that
$x_1,\ldots,x_s$ is the sorted order of $X_i$, but in
the improved algorithm in the next section the situation will be
more subtle.
The next lemma shows, that
the simple algorithm already provides an improvement of Aigner's
bound of $(k+\BigO(\sqrt k))n$.

\begin{lemma}[Complexity]
The number of comparisons of the simple algorithm
for $s=k$ on an $n$-element set
is $(k+\BigO(\log k))n + \BigO(k^3)$.
\end{lemma}

\begin{proof}
For processing the group $X_i$ in Step~\ref{alg1:sort}, we need
$\BigO(s\log s)+(k+1)(s-1)=k^2+\BigO(k\log k)$ comparisons, provided that
no restart is required. But since restarts may occur
only if the the oracle lies  at least once,
and the total number of lies is at most $k$, there are
no more than $k$ restarts for all groups together.
These restarts may account for at most $k(k^2+\BigO(k\log k))=
\BigO(k^3)$ comparisons. 
Thus, the total number of comparisons in Step~\ref{alg1:sort}
is $\frac ns(k^2+\BigO(k\log k))+\BigO(k^3)=
(k+\BigO(\log k))n+\BigO(k^3)$. 

As we mentioned in the introduction, 
the minimum (or maximum) of an $n$-element set
against $k$ lies can be found using $(k+1)n-1$ comparisons,
and so Step~\ref{alg1:find} needs no more than
$2(k+1)(n/s)=\BigO(n)$ comparisons. (We do not really need the optimal
algorithm for finding the minimum; any $\BigO((k+1)n)$ algorithm
would do.) The claimed bound on the total number of comparisons
follows.
\end{proof}

\section{The improved algorithm: Proof of Theorem~\ref{thm:1}}

In order to certify that $x_1$ is indeed the minimum of $X_i$,
we want that for every $x_j$, $j\ne 1$, the oracle declares
$x_j$ larger than some other element $k+1$ times. 
(In the simple algorithm, these $k+1$ comparisons
were all made with $x_{j-1}$, but any other smaller elements
will do.) This in itself requires $(k+1)(s-1)$ queries per group,
or $(k+1)(n-n/s)$ in total, which is already close to our target upper bound
in Theorem~\ref{thm:1} (we note that $s$ has to be at least
of order $k$, for otherwise, Step~\ref{alg1:find} of the
generic algorithm would be too costly).

Similarly, every $x_j$, $j\ne s$, should be
compared with smaller elements $k+1$ times, which again needs
$(k+1)(n-n/s)$ comparisons, so all but $\BigO(n)$ comparisons 
in the whole algorithm should better  be
used for \emph{both} of these purposes.

In the simple algorithm, the comparisons  used for sorting
the groups in Step~\ref{again} are, in this sense, 
wasted. The remedy is to use most of them
also for verifying the minimum and maximum in Step~\ref{verify}.
For example, if the sorting algorithm has already 
made comparisons of $x_{17}$ with 23 larger elements,
in the verification step it suffices to compare
$x_{17}$ with $k+1-23$ larger elements.

One immediate problem with this appears if the sorting
algorithm compares $x_{17}$ with some  $b>k+1$ larger elements,
the extra $b-(k+1)$ comparisons are wasted. However,
for us, this will not be an issue, because we will
have $s=k$, and thus each element can be compared
to at most $k-1$ others (assuming, as we may, that the sorting
algorithm does not repeat any comparison).

Another problem is somewhat more subtle. In order to explain
it, let us represent the comparisons made in the sorting
algorithm by edges of an ordered graph. The vertices are $1,2,\ldots,s$,
representing the
elements $x_1,\ldots,x_s$ of $X_i$ in sorted order,
and the edges correspond to the comparisons
made during the sorting, see the figure below on the left.

\iepsfig{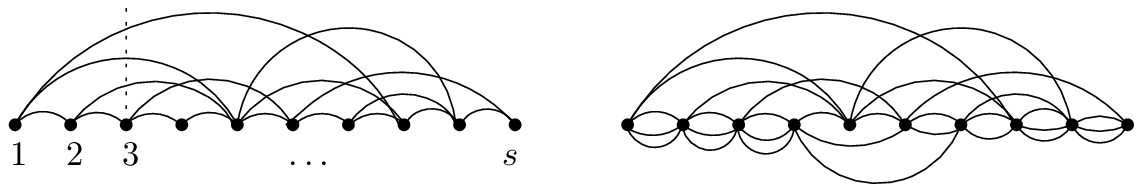}

In the verification step, we need to make additional comparisons 
so that every $x_j$, $j\ne 1$, has at least $k+1$ comparisons with
smaller elements and every $x_j$, $j\ne s$, has at least $k+1$
comparisons with larger elements. This corresponds to adding
suitable extra edges in the graph, as in the right drawing above
(where $k=2$, and the added edges are drawn on the bottom side).
  
As the picture illustrates, sometimes
we cannot avoid comparing some element with \emph{more}
than $k+1$ larger ones or $k+1$ smaller ones
(and thus some of the comparisons will be ``half-wasted'').
For example, no matter how we add the extra edges,
the elements $x_1,x_2,x_3$ together must participate in at
least $3$ half-wasted comparisons. Indeed, $x_2$ and $x_3$ together
require $6$ comparisons to the left (i.e.\ with a smaller element).  
These comparisons can be ``provided'' only by $x_1$ and $x_2$, which
together want only 6 comparisons to the right---but $3$ of these
comparisons to the right were already made with elements larger
than $x_3$ (these are the arcs intersecting the dotted vertical line
in the picture). 

The next lemma shows that this kind of argument is the only source
of wasted comparisons. For an ordered multigraph $H$ on
the vertex set $\{1,2,\ldots,s\}$ as above, let us define
$t(H)$, the \emph{thickness} of $H$, as $\max\{t(j): j=2,3,\ldots,s-1\}$,
where $t(j):=|\{\{a,b\}\in E(H): a<j<b\}|$ is the number of edges
going ``over'' the vertex~$j$.

\begin{lemma}\label{l:addedges}
Let $H$ be an undirected multigraph without loops on $\{1,2,\ldots,s\}$ such 
that for every vertex $j=1,2,\ldots,s$,
\begin{align*}
d_H^{{\rm left}}(j) & := |\{ \{i,j\}\in E(H) : i<j\}| \leq k+1\,,\\
d_H^{{\rm right}}(j) & := |\{ \{i,j\}\in E(H) : i>j\}| \leq k+1\,.
\end{align*}
Then $H$ can be extended to
a multigraph $\overline H$ by adding edges, so that
\begin{enumerate}
\item[\rm(i)] every vertex $j\ne 1$ has at least $k+1$
left neighbors and every vertex $j\ne s$ has at least $k+1$
right neighbors; and
\item[\rm(ii)] the total number of edges in $\overline H$
is at most $(k+1)(s-1) + t(H)$.
\end{enumerate}
\end{lemma}

The proof is a 
network flow argument and therefore constructive.
We postpone it to the end of this section.

For a comparison-based sorting algorithm $\mathcal{A}$,
we define the \emph{thickness} $t_\mathcal{A}(s)$ in 
the natural way: It is the maximum, over all $s$-element
input sequences, of the thickness $t(H)$ of the corresponding
ordered graph $H$ (the vertices of $H$ are ordered as in
the output of the algorithm and each comparison contributes
to an edge between its corresponding vertices). As the above lemma shows, 
the number of comparisons used for the sorting but not
for the verification can be bounded by the thickness of
the sorting algorithm.

\begin{lemma}\label{l:quicksort}
There exists a (deterministic) sorting algorithm $\mathcal{A}$
with thickness $t_{\mathcal{A}}(s)=\BigO(s)$.
\end{lemma}

\begin{proof}
The algorithm is based on Quicksort, but in order to control
the thickness, we want to partition
the elements into two groups of equal size in each recursive step.

We thus begin with computing the median of the given elements.
This can be done using $\BigO(s)$ comparisons (see, e.g., Knuth~\cite{Knuth1973};
the current best deterministic algorithm due to Dor and Zwick~\cite{DZ99}
uses no more than $2.95s +\littleo(s)$ comparisons). These algorithms
also divide the remaining elements into two groups,
those smaller than the median and those larger than
the median. To obtain a sorting algorithm, we simply
recurse on each of these groups.

The thickness of this algorithm obeys the recursion
$t_{\mathcal{A}}(s)\le \BigO(s)+t_{\mathcal{A}}(\lfloor s/2\rfloor))$,
and thus it is bounded by $\BigO(s)$.
\end{proof}

We are going to use the algorithm $\mathcal{A}$ from the lemma
in the setting where some of the answers of the oracle may be wrong.
Then the median selection algorithm is not guaranteed to partition
the current set into two groups of the same size and it is not
sure that the running time does not change. 
However, 
we can check if the groups have the right size and if the running
time does not increase too much. If some test goes wrong, we
restart the computation (similar to the simple algorithm).

Now we can describe the improved algorithm,
again by specifying Step~\ref{alg1:sort} of the generic
algorithm.

\medskip
\framedpar{
\begin{center}\bf Step~\ref{alg1:sort} in the improved algorithm
\end{center}
\begin{enumerate}
  \renewcommand{\theenumi}{\ref{alg1:sort}.\arabic{enumi}$'$}
\item \label{again2} (Sorting.) We sort the elements of $X_i$ 
by the algorithm $\mathcal{A}$ with thickness $\BigO(s)$
as in Lemma~\ref{l:quicksort}. If an inconsistency is detected
(as discussed above),
we restart the computation for the group $X_i$
from scratch.
\item \label{verify2} (Verifying the minimum and maximum.)
We create the ordered graph $H$ corresponding
to the comparisons made by $\mathcal{A}$, and 
we extend it to a multigraph $\overline{H}$
according to Lemma~\ref{l:addedges}.
We perform the comparisons corresponding to the added edges.
If we encounter an inconsistency, then we restart:
We go back
to Step~\ref{again2} and repeat the computation for the group $X_i$
from scratch. Otherwise,
we proceed with the next step. 
\item\label{allright2}
We set $m_i:=x_1$ and $M_i:=x_s$.
\end{enumerate}
}
\medskip

\begin{proof}[Proof of Theorem~\ref{thm:1}]
The correctness of the improved algorithm follows
in the same way as for the simple algorithm.
In Step \ref{verify2}, the oracle has declared
every element $x_j$, $j\ne 1$, larger than some other
element $k+1$ times, and so $x_j$ cannot be the minimum.
A similar argument applies for the maximum.

It remains to bound the number of comparisons.
From the discussion above, the number of comparisons is at most
$((k+1)(s-1)+t_{\mathcal{A}}(s))(\frac{n}{s}+k)+2(k+1)\frac{n}{s}$,
with $t_{\mathcal{A}}(s)=\BigO(s)$.
For $s=k$, we thus get that the number of comparisons at most
$(k+1+C)n+\BigO(k^3)$ for some constant $C$, as claimed.
\end{proof}

\begin{proof}[Proof of Lemma~\ref{l:addedges}]
We will proceed in two steps. First, we construct a multigraph $H^*$ from
$H$ by adding a maximum number of (multi)edges such that the left and right 
degree of every vertex are still bounded above by $k+1$. 
Second, we extend $H^*$
to $\overline{H}$ by adding an appropriate number
of edges to each vertex so that condition (i) holds.

For an ordered multigraph $H'$ on $\{1,2,\ldots,s\}$ with left
and right degrees upper bounded by $k+1$, let us define
the \emph{defect} $\Delta(H')$ as
\[
\Delta(H') := \sum_{j=1}^{s-1} (k+1-d_{H'}^{\mathrm{right}}(j))+
\sum_{j=2}^{s} (k+1-d_{H'}^{\mathrm{left}}(j))\,.
\]
We have 
$\Delta(H')=2(k+1)(s-1)-2e(H')$, where $e(H')$ is the number
of edges of~$H'$. 

By a network flow argument, we will show that
by adding suitable $m^*:= (k+1)(s-1)-e(H)-t(H)$
edges to $H$, one can obtain
a multigraph $H^*$ in which all left and right degrees
are still bounded by $k+1$ and such that $\Delta(H^*)
= 2t(H)$. The desired graph $\overline H$
as in the lemma will then be obtained  by adding
 $\Delta(H^*)$ more edges: For example, for every
vertex $j\ge 2$ of $H^*$ with $d_{H^*}^{\rm left}(j)<k+1$,
we add $k+1-d_{H^*}^{\rm left}(j)$ edges connecting $j$ to~$1$,
and similarly we fix the deficient right degrees
by adding edges going to the vertex~$s$.

It remains to construct $H^*$ as above. To this end,
we define an auxiliary directed graph $G$,
where each directed edge $e$ is also assigned an integral
capacity $c(e)$; see \figurename~\ref{fig:flow1}.

\begin{figure}[htbp]
  \centering 
  \subfigure[The graph $G$ with capacities.]{
  \label{fig:flow1}
  \includegraphics[scale=.95]{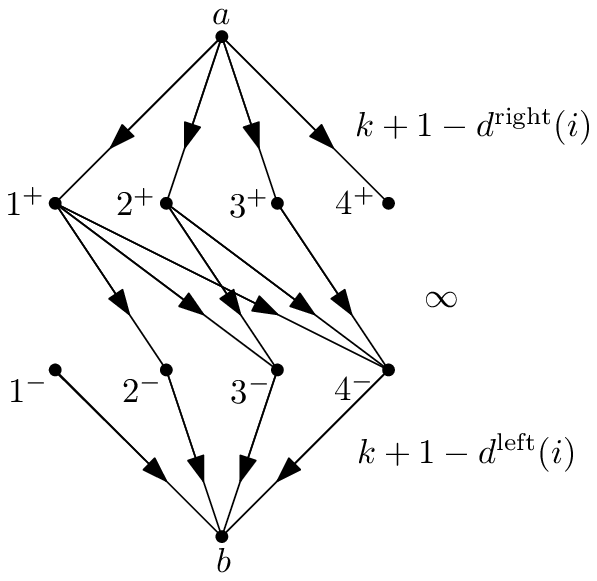}\hspace{0.5cm}}
  \subfigure[The cut $S_2$ in $G$.]{
  \label{fig:flow2}
  \includegraphics[scale=.95]{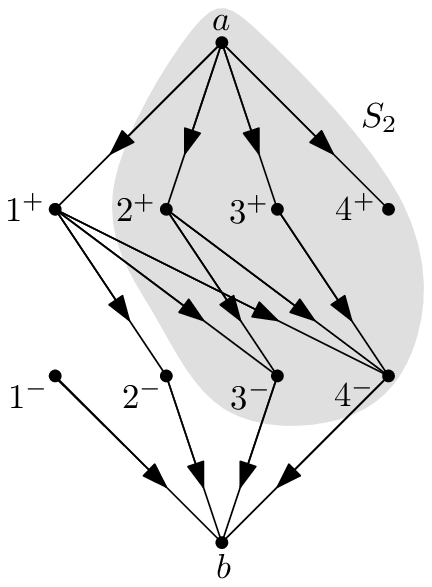}}
  \caption{The directed graph $G$ constructed in the proof of Lemma~\ref{l:addedges}.}
\end{figure}

The vertex set of $G$
consists of a vertex $j^-$ for every $j\in \{1,2,\ldots,s\}$, a
vertex $j^+$ for every $j\in \{1,2,\ldots,s\}$, and two special
vertices $a$ and $b$. There is a directed edge in $G$ from $a$
to every vertex $j^+$ and the capacity of this edge is
$k+1-d_H^{\mathrm{right}}(j)$. Similarly, there is a directed edge in $G$ from
every vertex $j^-$ to $b$, and the capacity of this edge is
$k+1-d_H^{\mathrm{left}}(j)$. Moreover, for every $i,j$
with $1\le i<j\le s$, we put the directed
edge $(i^+,j^-)$ in $G$, 
and the capacity of this edge is $\infty$ (i.e., a sufficiently
large number).

We will check that there is an integral $a$-$b$ flow in $G$
with value $m^*$ in $G$. By the max-flow min-cut theorem~\cite{FF56},
it suffices to show that every $a$-$b$ cut in $G$ 
has capacity at least~$m^*$ and there is an $a$-$b$ cut in $G$
with capacity $m^*$.

Let $S\subseteq V(G)$ be a minimum $a$-$b$ cut.
Let $i$ be the smallest integer such that $i^+\in S$. 
Since the minimum cut cannot use an edge
of unbounded capacity, we have $j^- \in S$ for all $j>i$. 

We may assume without loss of generality that  $j^+ \in S$ for all $j>i$ and $j^-\not\in S$ for
all $j\leq i$ (the capacity of the cut does not decrease by
doing otherwise). Therefore it suffices to consider
 $a$-$b$ cuts  of the form
\[
S_i := \{a\} \cup \{x^+ : x\geq i \} \cup \{x^- : x > i\}
\]
for $i=1,\ldots ,s$. The capacity of $S_i$, see \figurename~\ref{fig:flow2}, equals
\[
\sum_{j<i} c(a,j^+) + \sum_{j > i} c(j^-,b) = (s-1)(k+1) -\sum_{j<i} d_H^{\mathrm{right}}(j) - \sum_{j>i} d_H^{\mathrm{left}}(j)\,.
\]

Now let us look at the quantity 
$\sum_{j<i} d^{\mathrm{right}}(j) +\sum_{j>i} d^{\mathrm{left}}(j)$,
and see how much an edge $\{j,j'\}$ ($j<j'$) of $H$ contributes to it:
For $j<i<j'$, the contribution is $2$, while
 all other edges contribute~$1$. Hence the capacity
of the cut $S_i$ is $(k+1)(s-1)-e(H)-t(i)$, and
the minimum capacity of an $a$-$b$-cut is
$(k+1)(s-1)-e(H)-t(H)=m^*$ as required.

Thus, there is an integral flow $f$ with value $m^*$
as announced above. We now select the edges to be added to
$H$ as follows: 
For every directed edge $(i^+,j^-)$ of $G$, we add 
$f(i^+,j^-)$ copies of the edge $\{i,j\}$,
which yields the multigraph $H^*$. The number of added
edges is $m^*$, the value of the flow $f$, and
the capacity constraints guarantee that all left and
right degrees in $H^*$ are bounded by $k+1$. Moreover,
the defect of $H^*$ is at most $2t(H)$.
\end{proof}

\section{Concluding remarks}

We can cast 
the algorithm when $k=0$ sketched in the introduction
into the framework of our generic algorithm.
Namely, if we set $s=2$ and in Step~\ref{alg1:sort} we just
compare the two elements in each group, then we obtain
that algorithm.
The main feature of our algorithm is that every restart
only spoils one group.
This allows us to keep the effect of lies local.

In order to improve the upper bound of Theorem~\ref{thm:1}
by the method of this paper, we would need a sorting algorithm
with thickness $\littleo(s)$. (Moreover, to make use of the
sublinear thickness, we would need to choose $s$ superlinear in $k$,
and thus the sorting algorithm would be allowed to compare
every element with only $\littleo(s)$ others.)
The following proposition shows, however, that such a sorting algorithm
does not exist.
Thus, we need a different idea to improve Theorem \ref{thm:1}.

\begin{proposition}
Every (randomized) algorithm to sort an $s$-element set has thickness $\BigOmega(s)$ in expectation.
\end{proposition}

\begin{proof} 
By Yao's principle \cite{Yao},
it is enough to show that  every deterministic sorting 
algorithm $\mathcal{A}$ has expected thickness $\BigOmega(s)$ 
for a random input. In our case, we assume that the unknown linear
ordering of $X$ is chosen uniformly at random among all
the $s!$ possibilities.

In each step, the algorithm $\mathcal{A}$ compares some two elements
$x,y\in X$. Let us say that an element $x\in X$ is \emph{virgin}
at the beginning of some step if it hasn't been involved
in any previous comparison, and elements that
are not virgin are \emph{tainted}. A comparison is \emph{fresh}
if it involves at least one virgin element.

For notational convenience, we assume that $s$ is divisible by 8.
Let $L\subset X$ consist of the first $s/2$ elements in
the (random) input order (which is also the order of the output of
the algorithm), and let $R:=X\setminus L$.
Let $E_i$ be the event that the $i$th fresh comparison
is an \emph{$LR$-comparison}, i.e., a comparison in which
one of the two compared elements $x,y$ lies in $L$ and the other in $R$.
We claim that for each $i=1,2,\ldots,s/8$, the probability
of $E_i$ is at least $\frac13$.

To this end, let us fix (arbitrarily) the outcomes of all comparisons
made by $\mathcal{A}$ before the $i$th fresh comparison,
which determines the set of tainted elements, and let us also
fix the positions of the tainted elements in the input ordering.
We now consider the probability of $E_i$ \emph{conditioned}
on these choices.
The key observation is that the virgin elements
in the input ordering are still randomly distributed among the
remaining positions (those not occupied by the tainted elements).

Let $\ell$ be the number of virgin elements in $L$ and
$r$ the number of virgin elements in $R$; we have
$s/4\le \ell,r\le s/2$.

We distinguish two cases. First, let 
only one of the elements $x,y$ compared in the $i$th fresh comparison
be virgin. Say that $x$ is tainted and lies in $L$.
Then the probability of $E_i$ equals
$r/(\ell+r)\ge \frac13$.

Second, let both of $x$ and $y$ be virgin. 
Then the probability of $E_i$ is $2\ell r/((\ell+r)(\ell+r-1))$,
and since $s/4\le \ell,r\le s/2$, this probability exceeds~$\frac 49$. 

Thus, the probability of $E_i$
conditioned on \emph{every} choice of the outcomes of the
initial comparisons and positions of the tainted elements
 is at least $\frac 13$, and so the probability of $E_i$
for a random input is at least $\frac 13$ as claimed.
Thus, the expected number of $LR$-comparisons made by $\mathcal{A}$
is $\Omega(s)$. 

Let $a$ be the largest element of $L$, i.e., the $(s/2)$th element of $X$,
and let $b$ be the smallest element of $R$, i.e., the $(s/2+1)$st
element of $X$. Since we may assume that $\mathcal{A}$ doesn't
repeat any comparison, there is at most one comparison of $a$
with $b$. Every other $LR$-comparison compares elements
that have $a$ or $b$ (or both) between them. Thus, the expected
thickness of $\mathcal{A}$ is at least half of the
expected number of $LR$-comparisons, which is $\Omega(s)$.
\end{proof}

Note that the only thing which we needed in the proposition
above was that the corresponding ordered graph is simple and has
minimum degree at least $1$.

\section*{Acknowledgments}
We thank 
D\"om P\'alv\"olgyi for bringing the problem investigated in this
paper to our attention. This work has been started at the 7th Gremo
Workshop on Open Problems, Hof de Planis, Stels in Switzerland, July
6--10, 2009.  We also thank the participants of the workshop for the
inspiring atmosphere.

\end{document}